\documentclass[11pt,leqno]{amsart}

\usepackage{amsmath}
\usepackage{amsfonts}
\usepackage{amssymb}
\usepackage{amsthm,latexsym,mathrsfs}
\usepackage{graphicx,color}

\def \angle {\varphi}
\def \action {I}
\def \twist {\omega}
\def \pert {g}
\def \sgn {\mathrm{sgn}}

\def \A {\alpha}
\def \Z {\mathsf{z}}
\def \fund {\Phi}
\def \trans {T}
\def \Kes {\mu}
\def \pert {f}
\def \DS {\Delta S}

\def\I {\mathcal{I}}
\def\H {\mathscr H}

\newtheorem{theorem}{\sc{Theorem}}
\newtheorem{lemma}{\sc{Lemma}}
\newtheorem{proposition}{\sc{Proposition}}
\newtheorem{remark}{\sc{Remark}}

\oddsidemargin = - 0.1 cm
\allowdisplaybreaks
\date{\today}

\begin{document}

\title[Pinball dynamics]{Pinball dynamics: unlimited energy growth in switching Hamiltonian systems}

\author{Maxim Arnold}

\address{Maxim Arnold \\Department of Mathematics \\
University of Illinois \\
Urbana, IL 61801, U.S.A. and Institute for Information Transition Problems, Rus. Acad. Sci., Moscow, Russia.}

\email{mda@uiuc.edu}

\author{Vadim Zharnitsky}

\address{Vadim Zharnitsky \\ Department of Mathematics \\
University of Illinois \\
Urbana, IL 61801, U.S.A. }

\email{vzh@illinois.edu}

\begin{abstract}
A family of discontinuous symplectic maps on the cylinder is considered. 
This family arises naturally  in the study of nonsmooth Hamiltonian dynamics and 
in switched Hamiltonian systems.  The transformation depends on two parameters and 
is a canonical model for the study of  bounded and unbounded behavior in 
discontinuous area-preserving  mappings due to nonlinear resonances.
This paper provides a general description of the map and points out its connection with another map considered 
earlier by Kesten. In one special case,  an unbounded orbit is explicitly constructed. 
\end{abstract}
\maketitle
\section{Introduction.}

Theory of small perturbations of completely integrable Hamiltonian systems has a long  history 
that goes back to 19th century effort to explain  stability of planets. 
The  major breakthrough occurred in the late 1960s, when Kolmogorov-Arnold-Moser (KAM) theory was 
created.

The KAM theory states that under some non-degeneracy conditions,  stable motion persists  in a 
completely integrable Hamiltonian system under sufficiently small and smooth perturbation. 

For the original application, 3-body problem, the smoothness was not an issue as 
the gravity force is  analytic, outside of a small set of singularities. 
However, further applications of KAM theory to stability problems in physics and engineering, 
do require limited smoothness assumptions and also weaker forms of the so-called twist 
(nonlinearity) conditions. 

The degree of smoothness of the perturbation has a crucial role in the theory. 
In his famous ICM lecture Kolmogorov gave an outline of the theory where  he required  analyticity. 
Shortly,  V.I. Arnold proved Kolmogorov's statement,  also under the assumption  of analyticity. 
Independently, combining Kolmogorov's method with Nash smoothing technique, Moser proved a KAM type 
theorem requiring $333$  derivatives. 
Subsequently, the smoothness requirement was  reduced to single digits $(C^3)$ and several 
counterexamples have been found for lower regularity maps, see {\em e.g.} \cite{Herman}. 

Moser proved his theorem for the case of area-preserving monotone 
twists maps of the annulus. In this article we also restrict our attention to the 
representative case of twist maps on the plane, which corresponds to the periodically forced 
Hamiltonian systems with one degree of freedom. 

The above KAM counterexamples, that were constructed for the general twist maps,  do not provide a tool 
to decide stability in specific  physics problems. 
Therefore, it is important to investigate special maps arising in applications.

We note that even in the most extreme case of discontinuous maps, the stability problem is already nontrivial.
In the next section, we review  several such systems where boundedness problem for 
discontinuous maps naturally arises. Then we introduce a simple family of discontinuous twist maps, which  
captures the essential properties of those examples. The family contains a natural physical system which we 
call pinball transformation. The hallmark of the pinball map is the small twist, which on the one  hand 
frequently occurs in applications, and on the other hand makes stability problem rather delicate.  

In higher dimension, even the presence of KAM tori does not assure stability. The energy growth in smooth Hamiltonian systems in higher dimensions is an active area of research, see {\em e.g.} \cite{kaloshin, MF}.

\section{ Discontinuous twist maps and $\A \Z$-transformation.}

Discontinuous maps arise naturally in Hamiltonian systems with impacts, such as Fermi-Ulam problem, billiards, 
and more recently in hybrid or switched systems. It is usually the case that under the additional 
smoothness assumptions, KAM theory applies assuring boundedness of energy in all those problems.

One should also keep in mind that while the general monotone twist maps are characterized by a function 
of two variables $h(x_1,x_2)$,  these particular examples correspond to symplectic maps  characterized by 
function of one variable, {\em e.g.} for billiards $h(x_1,x_2) = ||x_2-x_1||$. Such a restriction 
makes it nontrivial to construct physically meaningful escaping trajectories.
 
For the readers' convenience, now we briefly  describe several such systems.

\subsubsection*{Example 1: Particle in  square wave switching potential}

Hybrid or switching systems  is an active area of research in applied mathematics and engineering sciences,
see {\em e.g.} \cite{LevKal,liberzon,ADDPR}. A prototype example of a switching system, where boundedness 
problem is already non-trivial, is a classical  particle in square wave periodic potential which changes 
the sign, periodically in time. 

More precisely, let the potential be  $V(x) = (-1)^{[x]}$ and assume 
the potential is switched every second  $V(x,t) = (-1)^{[t]}  \cdot (-1)^{[x]}$. While such potential 
is not differentiable, there is a natural way to define the dynamics by using the energy relation: the kinetic 
energy changes by 2 if the particle passes $t\in {\mathbb Z}$ integer points. It is common to ignore 
the singular subset of the extended phase space $(t,\dot x, x)$ where there is discontinuity in 
both time and space and the dynamics is not defined. Such subset has zero measure. Outside the singular set, the 
particle moves with constant speed $v=\sqrt{E\pm 1}$, gaining or losing energy by two at each switching, 
see the appendix for more details.

\subsubsection*{Example 2: Fermi-Ulam accelerator}
The Fermi-Ulam system consists of a classical particle bouncing between two 
periodically moving walls. The application of KAM theory shows that velocity (or energy) of 
the particle is uniformly 
bounded $|\dot v|< C(v(0))$, provided the periodically moving wall's position is sufficiently 
smooth $p(t) \in C^5(0,T)$, see \cite{laederich_levi}.

Fermi-Ulam problem can be reduced to a particle traveling in a periodic non-smooth potential
\[
\ddot x +  V^{\prime}(x,t) = 0, \,\,\, {\rm where} \,\,\,  V^{\prime} \in L^1({\mathbb T}^2). 
\]

It turns out that lack of  smoothness in $x$ ({\em e.g.} due to the presence of the wall in Fermi-Ulam problem) 
does  not destroy bounded behavior  as one can exchange the role of time and coordinate and then obtain a smooth 
monotone twist map by integrating over $x$, see {\em e.g.} \cite{levi}. 

If there is lack of smoothness in  both space and time in the periodic potential problem, then KAM theorems do not apply.

In the worst case the map is discontinuous, but even then, finding unbounded solutions could be challenging. 
One case, however, is more tractable: if jumps in the velocity (energy) are so large that the solution makes 
full revolution over 
one period of forcing so it will be in tune for the next velocity increase. A typical example would be given by this map 
\begin{equation}
\begin{cases}
x_1 = (x + y) \mod{1}\\
y_1 = y + \sgn(x_1-\frac 12).
\end{cases}
\end{equation}
Such scenario takes place in Fermi-Ulam problem if $p(t)$ has saw-tooth like shape. But, if the velocity increments are smaller,
then the twist will eventually detune the solution out of the resonance. 

\subsubsection*{Example 3. Outer Billiards.}
The question of boundedness becomes a lot more delicate and 
there are few  examples of escaping trajectories  for such systems in the dual billiards. Only recently, 
Schwartz and then Dolgopyat and Fayad constructed  unbounded solutions in the presence of  piecewise smooth boundary.
In the appendix, we give some heuristic description how our discontinuous twist map is related to  this example. \\

\section*{$\A\Z$--Map:  A model of boundedness problem for  discontinuous twist maps. }

\noindent
In this paper, we introduce a  two-parameter family of discontinuous
 monotone twist maps that seems to capture the essential difficulties of  several  switching-like (discontinuous) systems. 
The map is given by
\begin{equation}
\label{eq:AZ}
\begin{cases}
x_1 = (x + {\A} y^\Z ) \mod{1}\\
y_1 = y + \sgn(x_1-\frac 12)
\end{cases}
\end{equation}
and will be referred to as $\A \Z$-map, where $\A$ and $\Z$ are parameters. Note, that the map is invariant with respect to the natural scaling: varying the amplitude of the changes in the second variable or varying the length of the base circle in the first variable will lead to the equivalent system 
with different values of parameter $\A$.  We also observe that $\A\Z$ transformation preserves 
the unit-step lattice in action variable. In other words, the action variable is quantized for any fixed initial condition. 
 
For different values of parameters $\A\Z$ map corresponds to some natural systems:
\begin{itemize} 
\item $\Z$ = 1, \,\,\,   Fermi-Ulam with saw-tooth $p(t)$, discontinuous standard map.
\item $\Z$ = 0, \,\,\,   Erd\"os-Kesten  system (skew product of irrational rotation with jumps), which 
is defined in the next section. 
\item $\Z$ = 1/2, \,\,\, particle in switching square wave periodic potential.
\item $\Z$ = -1, \,\,\,  pinball problem, which is studied in this paper.
\end{itemize}

We explain in more details how $\A\Z$-transformation arises in each of these examples in the appendix.

\subsubsection*{Zero  twist example. Erd\"os-Kesten system.}
The following system was introduced by Erd\"os  and studied by Kesten \cite{Kesten} independently of 
any  KAM theory-type of problems. 
Erd\"os considered irrational rotation on the circle and asked what is the discrepancy between the orbit visiting 
different open subsets of the circle having equal measure. In particular, one can consider two halves of the circle $x\in (0,1/2)$ and $x \in (1/2,1)$.  In our notation, his system corresponds to the map with $\Z=0$.

In this  degenerate case, there is no twist in the system and the dynamics is a  skew product.  
Thus,  one can easily provide a set of values of parameter $\A$ (e.g.~$\A=1$) for which there are unbounded orbits. 
On the other hand for $\A=\dfrac 12$ any trajectory of the system \eqref{eq:AZ} is bounded since any point has period exactly $2$.  
In the generic case of irrational values of $\A$, Erd\"os' question leads to  an interesting number-theoretic problem. 
General result can be found in the paper by Kesten \cite{Kesten} where it is stated that for almost every $\A$ there is a set of
positive measure of orbits which escape to infinity but return to zero infinitely often. 
Most contemporary analysis of this phenomena can be found in \cite{Ralston}.     

Surprisingly, Erd\"os-Kesten (EK) system becomes important in the study of discontinuous twist maps after an appropriate renormalization procedure is carried out. 

\subsubsection*{Elementary properties of $\A\Z$-map} For non-degenerate twist $\Z\neq 0$ the following properties hold: 
\begin{itemize}
\item For any 
$\Z<-1$ nearly half of trajectories of the system \eqref{eq:AZ} escapes to infinity. It immediately follows from the fact that $\sum\limits_{n} n^\Z$ converges.

\item Fix $\Z\in \mathbb{N}$, then for $\A=1$, it is easy to verify that half the orbits are unbounded and for $\A=\frac 12$ 
all the trajectories are periodic.  

\item The most interesting and difficult problem of boundedness occurs for  $z\in [-1,1)$.
\end{itemize}

\section{Pinball system}
Now, we describe a simple mechanical system that corresponds to the case $z=-1$.
Consider now Fermi-Ulam like system with the fixed walls, but with one of the  walls containing a pinball mechanism: 
the momentum of the particle increases or decreases when it hits the wall according to the following law:
\begin{equation}
\begin{cases}
v\rightarrow v+1 \,\, {\rm if} \,\,   t\in [0,1/2]  \,\,   ({\rm mod} \,\, 1)   \\
v\rightarrow v-1 \,\, {\rm if} \,\,   t\in [1/2,1]  \,\,   ({\rm mod} \,\, 1),   
\end{cases}
\end{equation}
{\em i.e.} the momentum is increased (decreased) during the first (second) half period. 
This dynamics is described by the map with $\A$ being the distance between the walls and $\Z=-1$.

We rewrite the system \eqref{eq:AZ} for $\Z=-1$ and it will be called the pinball transformation that will be denoted by ${\bf P}$. 
For the sake of  clarity,  it would be more convenient to consider the base circle $\angle\in[0, 2)$.
\begin{equation}
\label{eq:pinball}
\begin{cases}
\angle_1=\left(\angle+\dfrac{\A}{\action}\right)\mod{2}
\\
\action_1=\action+\sgn(\angle_1-1)
\end{cases}
\end{equation}

Numerical experiments show that for typical values of parameter $\A$, the trajectory of the system \eqref{eq:pinball} is nearly recurrent 
for a long time and moreover approximates some piecewise smooth function having singularities only at the discontinuity lines 
$\angle=1$ and $\angle=0$.  

\begin{figure}[ht]
\centering
\includegraphics[height=2.5 in]{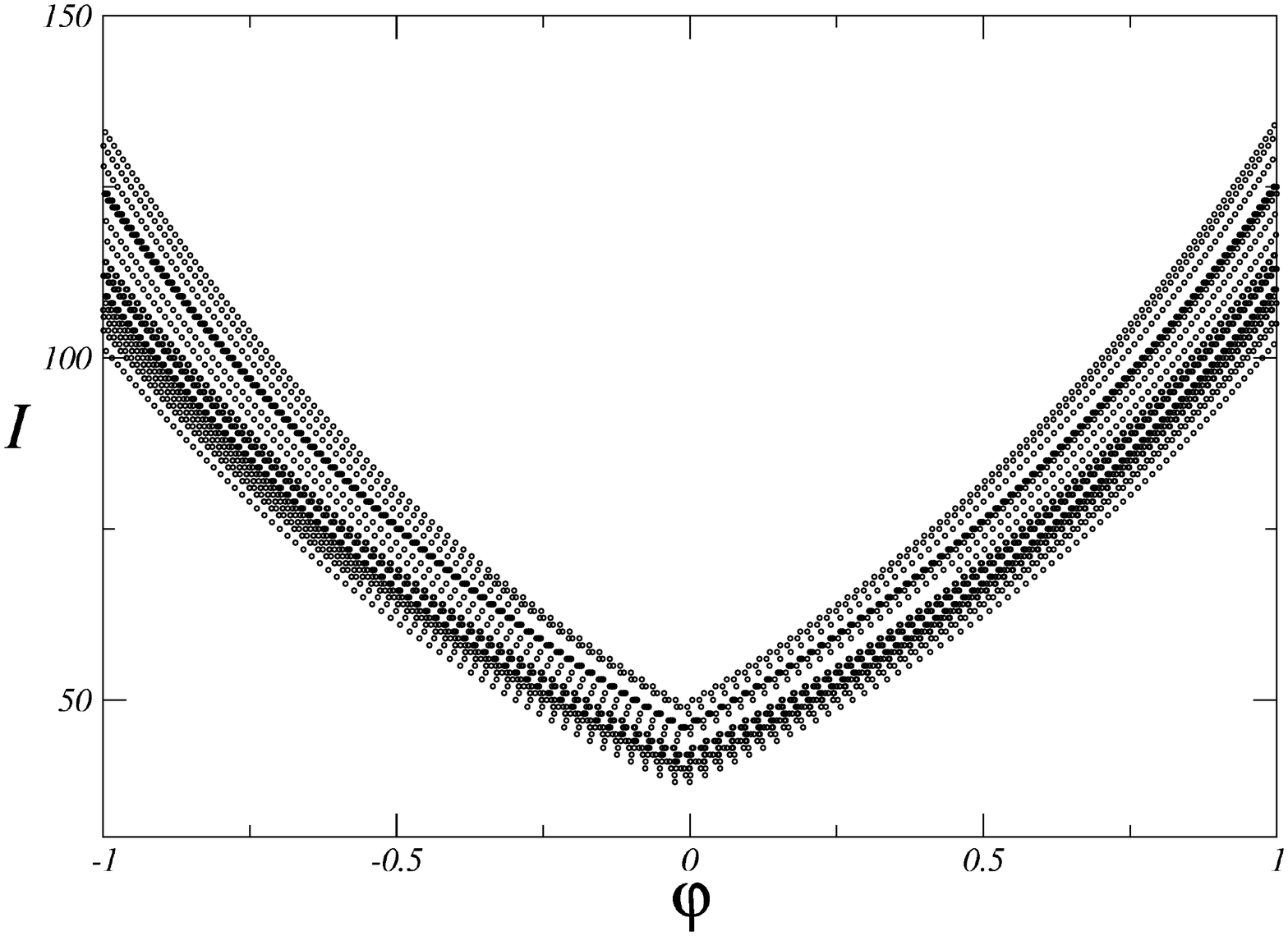}
\includegraphics[height= 2 in]{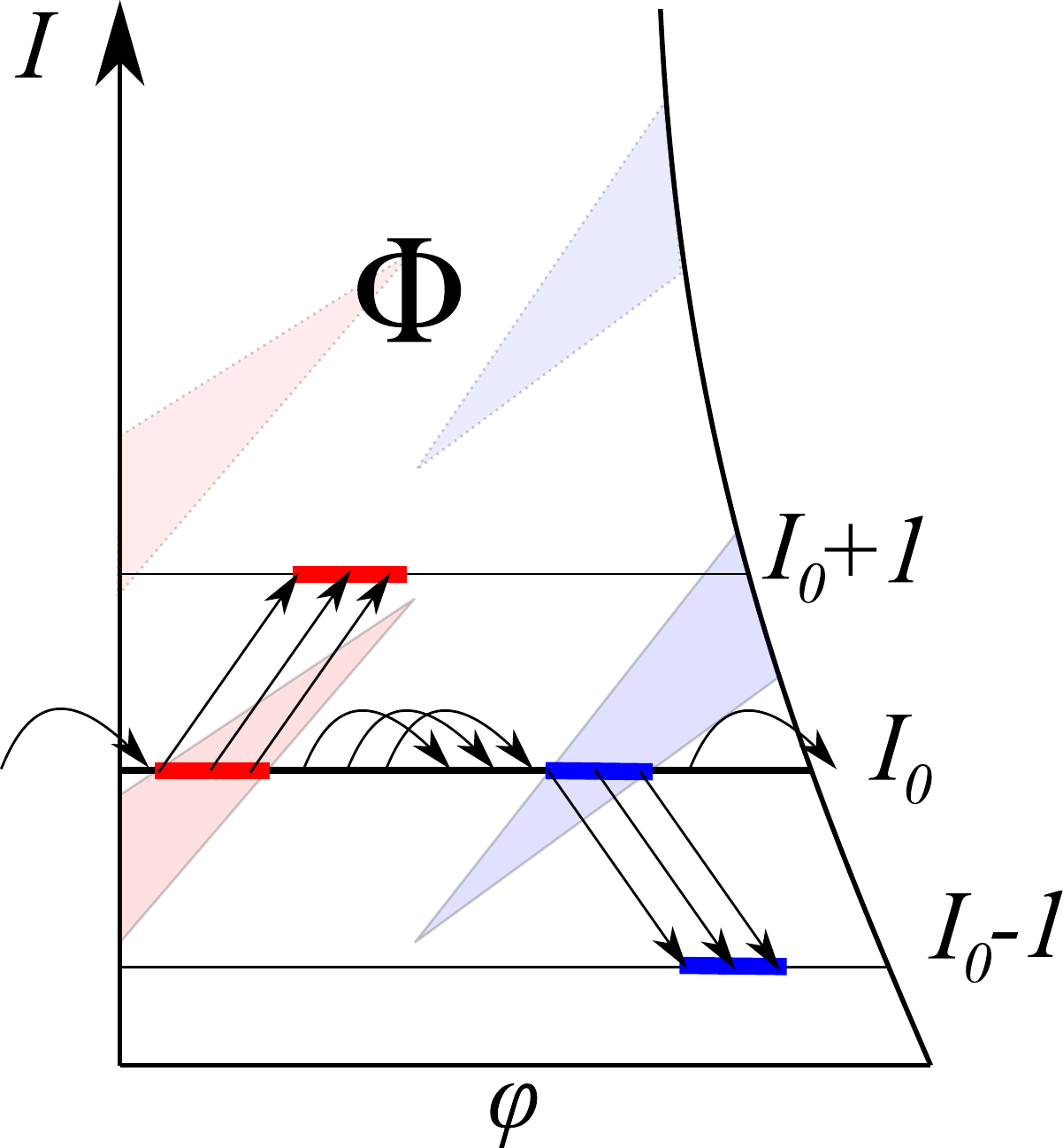}

\caption{Left: typical trajectory in a phase space: $10^6$ iterations starting at $(\angle_0,\,\action_0)=(0.01, 50)$.  Right: First return map into fundamental domain $\Phi$. Schematic description of domains of positive, neutral 
and negative growth.}
\end{figure}

Our goal is  to explain this  behavior by renormalizing the induced transformation in the so-called fundamental domain $\fund$. 
The fundamental domain is related to  Poincar\' e section for  flows in the sense that any orbit returns to it. In the Pinball map, we define 
 the fundamental domain as the set of points $(\angle,\action)$ located between singular line $\angle =0$ and its image
\begin{equation}
\fund = \{ (\angle,\action) \in [0,2) \times [0,\infty], 0< \angle < {\rm Proj_{\angle}}{\bf P}(0,I) \}.
\end{equation} 
The angular coordinate on $\fund$ will be denoted by $\tilde \angle \in [0,1]$.

Our first statement concerns the  asymptotic description of the first return  map:

\begin{theorem}
\label{th: Main} Let $\mu =\exp (\A^{-1})$ and assume $0 < \mu < 3$. Let $\chi_\mu^{(0)}$ and $\chi_\mu^{(1)}$ be 
characteristic functions of the $\dfrac 1\mu$ neighborhoods of the boundary of the fundamental domain, correspondingly, i.e. \mbox{$\chi_\mu^{(0)} = \chi_{[0,\mu^{-1}]}$} and $\chi_\mu^{(1)} = \chi_{[1-\mu^{-1}, 1]}$. 

Then, the  first return map of the domain $\fund$ under the map \eqref{eq:pinball} is a $O(\action^{-1})$ perturbation of the transformation 
\[\begin{cases}
\tilde \angle_1= \tilde \angle+\dfrac{2}{\mu} \, g_\mu(\action, \tilde \angle) \mod 1\\
\action_1=\action+\chi_+(\tilde \angle_1)-\chi_-(\tilde \angle_1)
\end{cases}\]
where  $[\cdot]$ denotes an integer part and \[g_\mu(\action,\angle)=\H_\mu(\action-1-\chi_\mu^{(1)}(\angle)) +\frac 12(1+\chi_++\chi_-)-(\chi_\mu^{(0)}+\chi_\mu^{(1)}) \] where
\[\H_\mu(x)=\chi_{\geqslant 0}(\mu + 2 \{ \mu x \} -3)-\{\mu x\}.\] 
Functions $\chi_+$ and $\chi_-$ are characteristic functions of two intervals $\I_+$ and $\I_-$ respectively. 
\end{theorem}


If $\mu$ is rational, then the map might possesses  a uniformly growing trajectory.  Indeed, if one of the periodic points stays longer in the positive part of the base interval than in  the negative part, then the corresponding trajectory grows without bound. 
Our construction of an escaping trajectory of the system \eqref{eq:pinball} consists in  choosing the initial data in an
appropriate way  so as to kill the leading order perturbation of the map. Next, we would have to estimate that the remaining perturbation
will not destroy such ``resonant'' growth. Combining these ideas, we prove

\begin{theorem}
\label{th: escape} For $\A=\dfrac{1}{\ln 2m}$, $m\in \mathbb{N}$ there exists 
an unbounded trajectory in the system \eqref{eq:pinball}.
\end{theorem}

\begin{remark}
If $\mu$ is irrational, then the leading order part of the map is reminiscent of EK map. Indeed, ignoring 
characteristic functions, the major part of the map takes the form
\[
\tilde \angle_1= \tilde \angle + \frac{2}{\mu} - \{\mu (I-1) \} \,\, {\rm or} \,\,  \tilde \angle_1= \tilde \angle + \frac{2}{\mu} - \{\mu (I-2) \}.
\]
It seems likely that the angular variable will be uniformly distributed and 
 one should expect similar behavior as found by Kesten. This will be the 
subject of future investigation.
\end{remark}

\begin{remark}
Note that smoothing the signum function  discontinuity in \eqref{eq:pinball} will make KAM theory applicable and then all 
solutions will be bounded. 
\end{remark}

Indeed,  change the variables: $(\angle,y)=(\angle, \action^{-1})$. 
In the new variables, the smooth version of the Pinball transformation takes the form
\[\begin{cases}
\angle_1 = \angle +\A y\mod 2 \\
y_1=\dfrac{1}{\frac{1}{y} - f(\angle,y)}
\end{cases}\]
where $f$ is smooth and $|f(\angle,y)|<1$. Then 
\[
y_1=y\dfrac{1}{1 - y f(\angle,y)}=y\left(1 + \sum_{n=1}^{\infty} (yf(\angle,y))^n\right)
\]
so the perturbation is of order $O(y^2)$ which is much smaller than the twist. The curve intersection property 
follows from the area-conservation in the original variables. Therefore, this map satisfies the conditions of monotone 
twist theorem, see {\em e.g.} \cite{MF}.

\section{Proof of Theorem \ref{th: Main}.}

Recall the definition of the fundamental domain as a subset $\fund\subset \mathbb{S}\times \mathbb{R}_+$ between the discontinuity line $\phi=0$ and its first iteration: 
\begin{equation}
\label{eq:fundamental_domain}
\fund=\left\{(\angle,\action)\mid \angle\in \left(0,\dfrac{\A}{ \action-1}\right)\right\}
\end{equation} 
and consider the transformation $\trans(\angle,\action)=(\angle', \action')$
as the first return map  for any point $(\angle,\action)\in \fund$ according to \eqref{eq:pinball}.
We have the following bound on the action change

\begin{lemma}
\label{lm: First_return}
If $(\angle', \action')$ is an image of the point $(\angle, \action)$ under the transformation $\trans$ then $|\action'-\action|\leqslant 1$.
\end{lemma}

In other words, Lemma \ref{lm: First_return} states that as the angle variable winds around the cylinder and 
the action variable undergoes large changes, after returning to the fundamental domain,  the action will not change by more than 1.
This property assures a good local control on the orbits.

\begin{figure}[ht]
\begin{center}
\includegraphics[width = 0.25\textwidth]{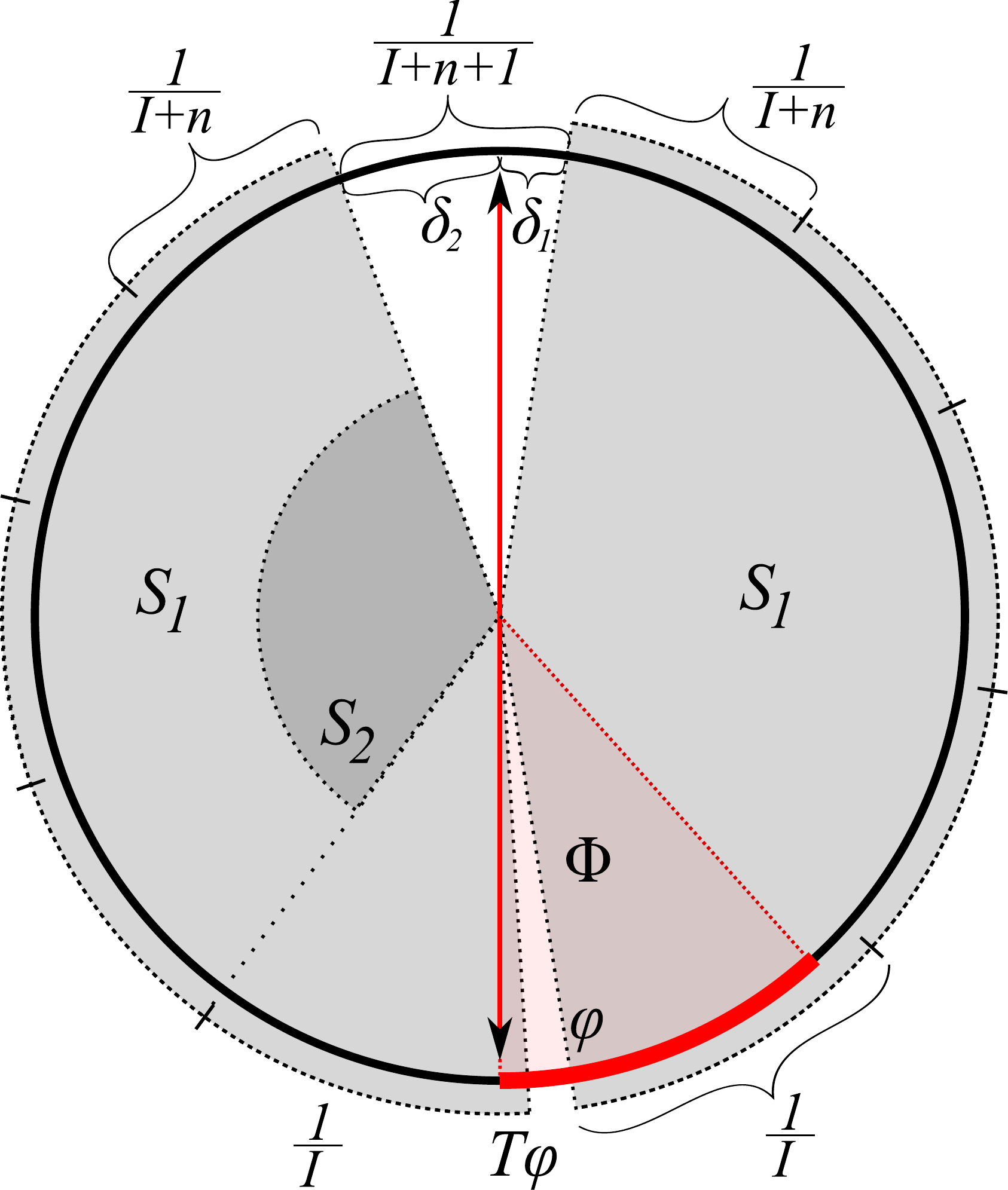}
\hspace{0.1\textwidth}
\includegraphics[width = 0.55\textwidth]{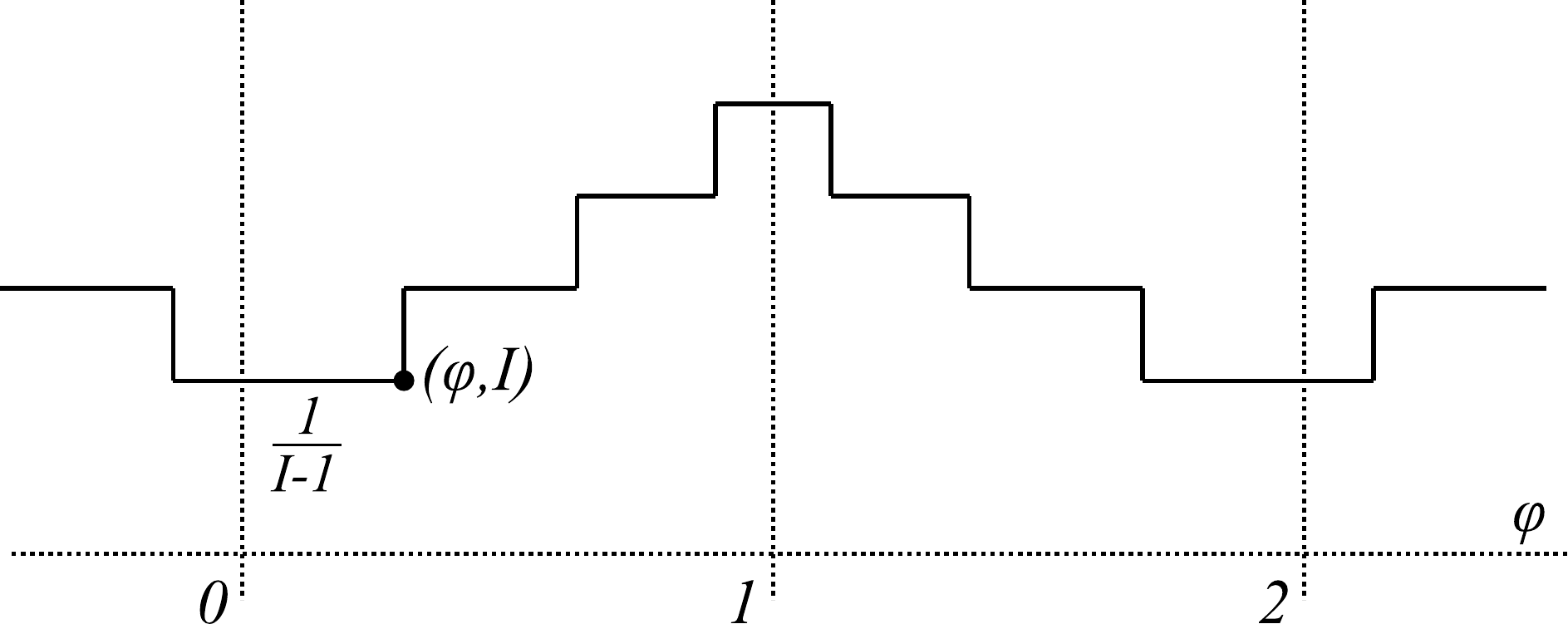}
\end{center}
\caption{Mapping from fundamental domain to itself represented on the base circle 
(left panel) and on the covering space (right panel).}
\label{fig:circle}
\end{figure}

Next, we describe  the structure of the subsets  in the  base $\angle\in (0,\A/(I-1))$ for which the action increases $\action'=\action+1$ or decreases 
$\action'=\action-1$.

Let us introduce the $\emph{rescaled}$ angle variable $\tilde{\angle}=(\action-1)\angle \A^{-1}$. In the renormalized variables, the fundamental domain can be represented by 
 \[
\fund=(\tilde{\angle},\action)=[0,1]\times \mathbb{R}_+.
 \]   
 
\begin{lemma} \label{lm:positive-negative} 
The set $\fund \cap \{\action=C\}$ is the union of three disjoint subsets 
\[
\fund\cap \{\action=C\}=\I_+\cup \I_- \cup \I_0, 
\]
where $\I_+$ and $\I_-$ are intervals of equal measure and consist of all points $(\angle,C)$ such that $T(\angle, C)=(\cdot,C\pm 1)$ respectively. 

The intervals are contained in the regions
\[
\I_+ \in \left ( 0,\frac{\A}{\mu (I-1)} \right )   \,\,\,  {\rm and } \,\,\, 
\I_- \in \left (  \frac{\A(1-\mu^{-1})}{I-1}, \frac{\A}{I-1} \right ),
\]
where $\Kes(\A) = \exp(\A^{-1})$.

 \end{lemma}
 
 Finally, we have the following lemma which ends the proof of Theorem \ref{th: Main}.
 
 \begin{lemma}
 \label{lm: renormalization_Kesten}
 The first return map $\trans$ in the rescaled variables takes the form
 \begin{equation*}
 T(\angle,\action)=\begin{cases}
 \tilde{\angle}_1=\tilde{\angle}+\dfrac{2}{\mu} \,\, g_\mu(\action,\tilde{\angle}) + O(\action^{-1})\\
 \action_1=\action+\chi_+(\tilde{\angle}_1)-\chi_-(\tilde{\angle}_1)
 \end{cases}
 \end{equation*} where    $\Kes(\A) = \exp(\A^{-1})$,  $\chi_{\pm}$ are characteristic functions of positive and negative intervals depending 
weakly on $I$.

 \end{lemma}

 \vspace{5mm}

\section{Proofs}

\subsection*{Proof of Lemma \ref{lm: First_return}}
\begin{proof}
We begin with giving a heuristic argument based on the Figure \ref{fig:circle} 
(right panel). 
We can represent orbits as stairs going up and down in the domain $\angle \in [0,2]$. Consider the special configuration that is symmetric with respect to $\angle = 1$. It is easy to see that the corresponding solution
will have the same action after the first return map. By moving the graph so that the top level does 
not cross $\angle=1$  it is possible to see that the change in the action cannot be more than one.
Indeed, the lowest level steps are wider than the top one and therefore at most one crossing can occur.

Now, we provide the full proof.
Assume  the initial point $(\angle,\action)$ is in the fundamental domain $\fund$, {\em i.e.} $\angle\in (0,\A/(\action-1))$.  The map \eqref{eq:pinball} 
is iterated $n+1$ times while $\action$ increases until the orbit is one step away from crossing $\angle = 1$. Next, the map is iterated  $n'+1$ times while $\action$ decreases until the orbit is one step away 
from crossing $\angle =2$.  Then, by the definition of the fundamental domain, last $(n+n'+2) $-nd iterate is in  $\fund$ and we obtain

\begin{equation}
\label{eq:first_return1}
\trans(\angle,\action)=\left(\angle+\A S_1(\angle,\action) + \A S_2(\angle,\action)+\dfrac{\A}{\action+n+1}+\dfrac{\A}{\action+(n-n'-1)}, \action+n-n'\right),
\end{equation}
where
\begin{equation}
\label{eq:S_12}
S_1(\angle,\action)=\sum\limits_{k=0}^{n} \dfrac {1}{\action+k},\qquad
S_2(\angle,\action)=\sum \limits_{k=0}^{n'}\dfrac {1}{\action+n-k}.
\end{equation}
The numbers $n$ and $n'$ are uniquely defined by the relations:
\begin{equation}
\label{eq:sk1}
\begin{cases}
1>\angle+\A S_1(\angle,\action)\\
1<\angle+\A S_1(\angle,\action)+\dfrac{\A}{\action+n+1}
\end{cases}
\end{equation}
and
\begin{equation}
\label{eq:sk2}
\left\{
\begin{array}{cl}
2>&\angle+\A S_1(\angle,\action) + \A S_2(\angle,\action)+\dfrac{\A}{\action+n+1}\\\\
 2<&\angle+\A S_1(\angle,\action) + \A S_2(\angle,\action)+\dfrac{\A}{\action+n+1}+\dfrac{\A}{\action+n-n'-1}.
\end{array}
\right.
\end{equation}

The equation  \eqref{eq:sk1} means that $n$-th iterate is the last one staying in the  right half-circle, so the next iterate will be in the left half-circle. Similarly $(n+1)+n'$-th iterate is the last one  before returning to $\fund$, see Figure \ref{fig:circle}.

 Denote by $\DS=S_2-S_1$. Rewriting the second sum in \eqref{eq:S_12} for $k'=n-k$ we obtain for $\DS$:
\begin{equation}
\label{eq:S'}
\DS=\sum\limits_{k'=n-n'}^{n}\dfrac{1}{\action+k'}-\sum\limits_{k=0}^{n}\dfrac{1}{\action+k}.
\end{equation}

Note that the expression \eqref{eq:S'} implies that 
\begin{equation}
\label{eq:S_cont}
\begin{cases}
\DS>\dfrac{1}{\action-1}+\dfrac{1}{\action-2} &\mbox{ if } n-n'<-1
\\
 \DS<-\dfrac{1}{\action} -\dfrac{1}{\action+1} &\mbox{ if } n-n'>1.
\end{cases}
\end{equation}
To prove that $|n-n'|\leqslant 1$ it is sufficient to show that \eqref{eq:S_cont} contradicts  \eqref{eq:sk1} --
\eqref{eq:sk2}. Rewriting \eqref{eq:sk2} for  $S_2=\DS+S_1 $ we obtain 
\begin{equation}
\label{eq:ds1}
2>\angle+2\A S_1(\angle,\action)+\A \DS+\frac{\A}{\action+n+1}>2-\frac{\A}{\action+n-n'-1}.
\end{equation} 
Multiplying the equation \eqref{eq:sk1} by $2$ and
substituting  into \eqref{eq:ds1} we get 

 \begin{equation}
 \label{eq:s*}
 \left\{
 \begin{array}{cl}
 \DS&<\dfrac{\angle}{\A}+\dfrac{1}{\action+n+1}
 \\\\
 \DS&>\dfrac{\angle}{\A}-\dfrac{1}{\action+n+1}-\dfrac{1}{\action+n-n'-1}.
 \end{array}
 \right.
 \end{equation}

Since for any point $(\angle, \action)\in \fund$ one has $\dfrac{\angle}{\A}\in\left(0,\dfrac{1}{ \action-1}\right)$ inequalities \eqref{eq:s*}  contradict  \eqref{eq:S_cont}.

\end{proof} 

\subsection*{Proof of lemma \ref{lm:positive-negative}}
\begin{proof}

 As already found in the proof of Lemma \ref{lm: First_return}, transformation $\trans: \fund \rightarrow \fund$ in $\angle$-variable takes the form 
 \begin{equation}
 \label{eq:First_return}
\angle'=\angle+\left(2\A S_1(\angle,\action)+ \dfrac{\A}{\action+n+1}-2\right)+\A\Delta S,
\end{equation}
 where 
 \begin{equation}
 \label{eq:DeltaS}
 \Delta S=
 \begin{cases}
 0,  &    {\rm if}  \,\, \angle \in \I_+\\
 \dfrac{1}{\action-1},&  {\rm if}  \,\, \angle \in \I_0\\
 \dfrac{1}{\action-1}+\dfrac{1}{\action-2},&  {\rm if} \,\,  \angle \in \I_-\\
 \end{cases}
 \end{equation} 
Define $\delta_1, \delta_2$ by the relations (see Fig. \ref{fig:circle})
   \[\angle+\A S_1(\angle,\action)=1-\delta_1\]
   \[\angle+\A S_1(\angle,\action)+\dfrac{\A}{\action+n+1}=1+\delta_2,\]
then it is easy to see that
   \[\angle+2\A S_1(\angle,\action)+\dfrac{\A}{\action+n+1}-2=\delta_2-\delta_1-\angle.\]

Then the expression \eqref{eq:DeltaS} could be rewritten in the form
   \begin{equation}
   \label{eq:I+-0}
   \begin{cases}
   \delta_2-\delta_1-\angle>0,& \mbox{for}\quad\angle\in\I_+\\
   \delta_2-\delta_1-\angle \in \left(-\dfrac{\A}{\action-1},0\right),&\mbox{for}\quad\angle\in\I_0 \\
    \delta_2-\delta_1-\angle<-\dfrac{\A}{\action-1},&\mbox{for}\quad\angle\in\I_-\\
   \end{cases}
   \end{equation}

 In particular, the first inequality in \eqref{eq:I+-0} implies  that for 
 $\angle\in \I_+$ one has 
\[
\angle<\dfrac{\A}{\action+n+1}, \,\, {\rm since} \,\,\,\,  \dfrac{\A}{\action+n+1}=\delta_2+\delta_1.
\]
    Similarly, it is easy to verify that  if $\angle\in \I_-$ then $\angle >\dfrac{\A}{\action-1}-\dfrac {\A}{\action+n+1}$. 
The next proposition describes some rigidity properties  of the intervals $\I_0,\I_+,\I_- $ in   \eqref{eq:I+-0}

\begin{proposition}
\label{lm:rigid}
Assume that $(\angle_0, \action)\in \I_+$ and that the corresponding sum $S_1(\angle_0,\action)$ has $n$ terms. Then the  point $(\angle_\epsilon=\angle_0+\epsilon, \action)$ belongs to $\I_+$  if and only if the sum $S_1(\angle_\epsilon,\action)$ has also $n$ terms. 
\end{proposition}
\begin{proof}
Denote by $\delta_1^{(0)}$ and $\delta_2^{(0)}$ the corresponding parts of the interval $\dfrac{\A}{\action+n+1}$ for the point $\angle_0$. By assumption $\angle_0\in \I_+$ and 
thus $\delta_2^{(0)}-\delta_1^{(0)}-\angle_0>0$, then by 
\eqref{eq:I+-0} it is sufficient to verify that 
$\delta_2^{(\epsilon)}-\delta_1^{(\epsilon)}-\angle_\epsilon$ remains positive. 

Then the proof can be obtained from  the following calculation: \\

{\bf (a)} First, consider the case when the number of terms remains the same (equal to $n$).
Then $\delta_2^{(\epsilon)}=\delta_2+\epsilon$, $\delta_1^{(\epsilon)}=\delta_1-\epsilon$ and so 

\vspace{3mm}

\begin{equation}
\label{eq:rigid}\delta_2^{(\epsilon)}-\delta_1^{(\epsilon)}-\angle=(\delta_2^{(0)}+\epsilon) -(\delta_1^{(0)} - \epsilon) -\epsilon=\delta_2^{(0)}-\delta_1^{(0)}+\epsilon>\delta_2^{(0)}-\delta_1^{(0)}
\end{equation}

\vspace{3mm}

{\bf (b)} Assume now that the number of terms in $S_1(\angle_\epsilon,\action)$ is different from 
$n$. Assume it contains $n-1$ terms (all other cases can be treated similarly). Then we have, see Figure \ref{fig:sdvig}.

\begin{figure}[ht]
\centering
\includegraphics[width=0.3\textwidth]{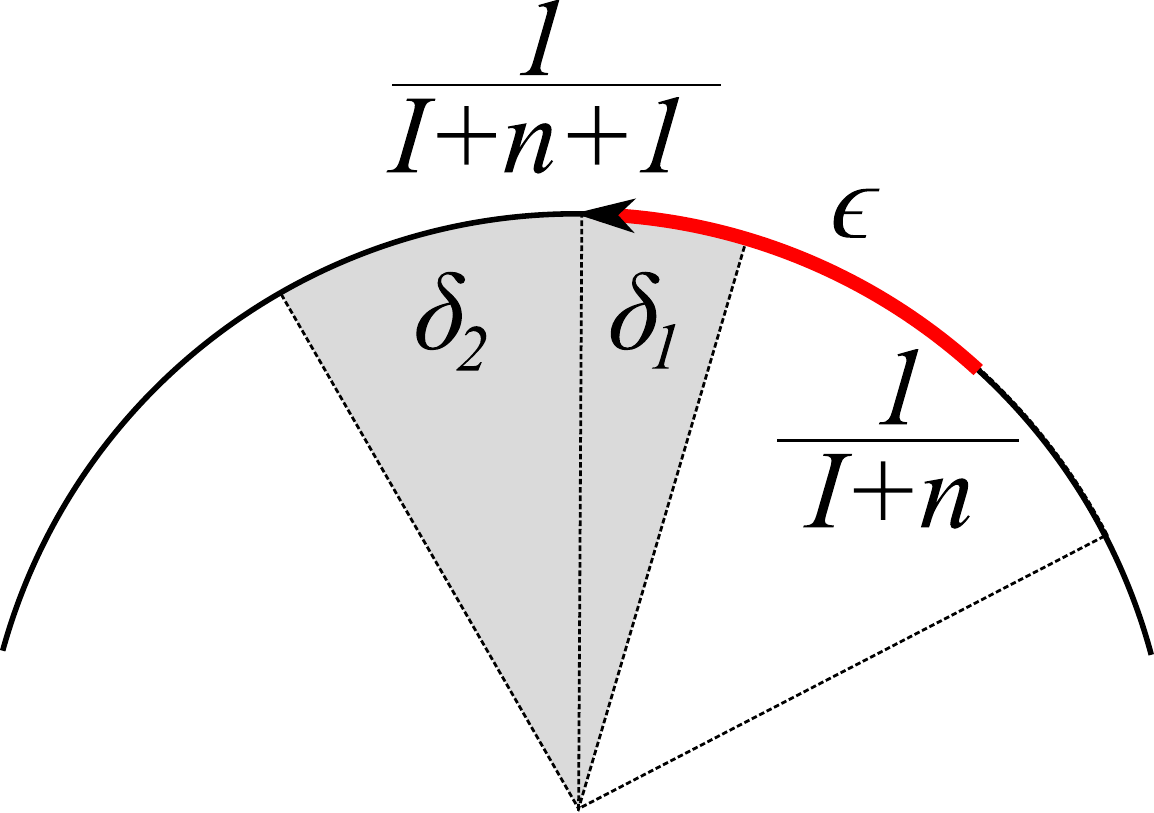}
\caption{Rotation by $\epsilon$.}
\label{fig:sdvig}
\end{figure}

\[
\delta_2^{(\epsilon)}=\epsilon-\delta_1^{(0)}=\epsilon+\delta_2^{(0)}-\dfrac{\A}{\action+n+1}
\]
which implies
\[
\delta_1^{(\epsilon)}=\dfrac{\A}{\action+n}-\delta_2^{(\epsilon)}=\dfrac{\A}{\action+n}-\epsilon+\delta_1^{(0)}.
\] 
Finally, we have 
\[\delta_2^{(\epsilon)}-\delta_1^{(\epsilon)}-\angle_\epsilon=\epsilon+\delta_2^{(0)}-\delta_1^{(0)}-\angle_0-\dfrac{\A}{\action+n+1}-\dfrac{\A}{\action+n}.\]
The latter expression is negative since by construction 
\[
\delta_2^{(0)} + \epsilon < \frac{1}{I+n} + \frac{1}{I+n+1}.
\]
This ends the proof of the proposition.
\end{proof}
Using proposition \ref{lm:rigid}, we will find the set $\I_+$. Let the initial angle be $\angle_0=0$ then 
\[
\delta_2^{(0)}~=~\left\{\A\sum\limits_{j=0}^{n+1}\dfrac{1}{\action+j}\right\} \,\, {\rm and} \,\, 
\delta_1^{(0)}=\dfrac{\A}{\action+n+1}-\delta_2^{(0)}.
\]
If $\delta_2^{(0)}>\delta^{(0)}_1$ then  by \eqref{eq:I+-0} the point $(0,\action)$ belongs to $\I_+$ and then  proposition \ref{lm:rigid} 
implies that $\I_+=\left(0,\delta_1^{(0)}\right)$. 

Otherwise, if $\delta_2^{(0)} < \delta^{(0)}_1$ then  $(0,\action)$ does not belong to $\I_+$ and to find the leftmost point 
$(\phi = \epsilon,\action)\in \I_+$ we need to satisfy the condition $\delta_2^{(\epsilon)}-\delta_1^{(\epsilon)}-\epsilon>0$. 
Using \eqref{eq:rigid} we obtain  $\epsilon=\delta_1^{(0)}-\delta_2^{(0)}$ and therefore in this case 
$\I_+=\left (\delta_1^{(0)}-\delta_2^{(0)}, \delta_1^{(0)}\right )$.

Similar calculations can be carried out  for $\I_-$, but  it is easier 
to use symmetric properties of the map with respect to reversing the ``time''.

Then, if  $\delta_1^{(00)}>\delta^{(00)}_2$ ($\delta_i^{(00)}$ correspond to $\delta^{(0)}$ but are 
obtained from negative iterates) the point $(0-,I-1)$ which is the preimage of  
$(\frac{\A}{I-1}-,I)$ belongs to $\I_-$ and by the same argument as above for $\I_+$, we have
\[
\I_-=\left ( \frac{\A}{\action-1}-\delta_2^{(00)},\frac{\A}{\action-1}\right ).
\]

If the converse holds, {\em i.e.} $\delta_2^{(00)}>\delta^{(00)}_1$ we obtain 
\[
\I_-=\left ( \frac{\A}{\action-1}-(\delta_2^{(00)}-\delta_1^{(00)}), 
\frac{\A}{\action-1}-\delta_2^{(00)} \right ). 
\]
This ends the proof of the lemma.
\end{proof}
The location and the measure of the intervals $\I_+ $ and $\I_- $ are then controlled by the fractional 
parts of the sums $\A S_1(0,\action)$. We shall estimate these quantities in the proof of the next lemma.

 \subsection*{Proof of Lemma \ref{lm: renormalization_Kesten}}
\begin{proof}

For the clarity of presentation, we give the proof only for the case  $\action=N$. The general case can be treated similarly.  
We introduce the notation for the expression
\[
\A \psi(\angle, \action)= 2\A S_1(\angle,\action)+\frac{\A}{\action+n+1}-2 
\] 
from \eqref{eq:First_return}.

  In the rescaled variables, the  transformation \eqref{eq:First_return} takes the form 
  \begin{equation}
  \label{eq:rescaling}
  \tilde{\angle}'=
  \begin{cases}
  \frac{\action}{\A}\angle'=\left(1+\dfrac{1}{\action-1}\right)\tilde{\angle}+\action\psi(\action,\angle),&\angle\in \I_+
  \\
  \frac{\action-1}{\A}\angle'=\tilde{\angle}+(\action-1)\psi(\action,\angle)+1,&\angle\in \I_0
  \\
  \frac{\action-2}{\A}\angle'=\left(1-\dfrac{1}{\action-1}\right)\tilde{\angle}+(\action-2)\psi(\action,\angle)+2-\dfrac{1}{\action-1},&\angle\in \I_-
   \end{cases}
  \end{equation}

\noindent 
{\bf Case $\angle\in \I_+$.} If the action variable $I$ increases, then the  
number $n$ of iterates of the transformation \eqref{eq:pinball} contained in the right half-circle 
(where the action is gained) is greater than the number $n'$ of iterates contained  in the left half-circle.
  Lemma \ref{lm: First_return} implies that $n=n'+1$. Thus, the transformation $\trans$ for these values of $\angle$ takes the form
 \[\angle'=\angle+2\A S_1(\angle,\action) +\dfrac{\A}{\action+n+1}-2,\]
where
\[S_1(\angle, N) = \sum\limits_{k=0}^n \dfrac{1}{N+k}
,\]

\begin{proposition} 
\label{lm:n} Assume that $\mu \in (1,3)$, then  for $\angle \in \I_+$ the number $n$ is given by  
\begin{equation}
\label{eq:right_n}
n = \mu(N-1) - N + \H_\mu(N-1),
\end{equation} where $\mu= \exp(1/\A)$ and 
\[\H_\mu(x)= \chi_{\geqslant 0}(\mu + 2 \{ \mu x \} -3)-\{\mu x \}.\]
\end{proposition}
\begin{remark}
If $\mu >3$, then an additional term $ \left[\frac {\mu-1}{2}\right]$ has to be added in \eqref{eq:right_n}. 
We omit the detailed calculations of this more general case, which can be reproduced in the same way. 
\end{remark}

\begin{proof}
 By proposition \ref{lm:rigid}, the number of steps
 in the positive half of the cylinder  $n$ is independent of  $\angle\in \I_+$, therefore   it is sufficient 
to verify  \eqref{eq:right_n}  only for some $\angle\in \I_+$ using \eqref{eq:sk1}. Moreover, 
we can take $\angle=0$ even though this point might not be in $\I_+$. Indeed, according 
to lemma \ref{lm:positive-negative} the point $\angle=0$ satisfies 
$\angle\in \I_+$ or $\angle\in \I_0$. If the former, we are done. If the latter then 
it is possible to see that by increasing $\angle$ we will eventually cross into $\I_+$ without 
changing $n$ (the number of steps in positive part of the cylinder).

We have 

\[
S_1(\angle, N) = \sum\limits_{k=0}^{n} \dfrac{1}{N+k} = H_{N+n}-H_{N-1},
\]

where $H_k$ denotes $k$-th harmonic number. The harmonic number has the following asymptotic expansion
\[H_k=\ln k +\gamma +\dfrac{1}{2k}-\sum\limits_{j=1}^{\infty}\dfrac{B_{2j}}{2jk^{2j}}\]
where $B_{j}$ denotes $j$-th Bernoulli number. Therefore, 
\begin{equation}
\label{eq:S_N^k}
\begin{split}
S_1(\angle,N)=\ln \dfrac{N+n}{N-1} &+\dfrac{1}{2}\left(\dfrac{1}{N+n}-\dfrac{1}{N-1}\right)-\dfrac{1}{12}\left(\dfrac{1}{(N+n)^2}-\dfrac{1}{(N-1)^2}\right)+O\left(\dfrac{1}{N^4}\right)
\end{split}
\end{equation}

We will take $n = [\mu(N-1)] - N + \gamma$ and will verify that $\gamma$ is given by the above Heaviside function.

Using the relation $[\mu(N-1)]=\mu(N-~1)-~\{\mu(N-~1)\}$ and denoting  
$\ell:=\dfrac{\gamma}{\mu(N-1)}$ 
and  $y:=~\dfrac{\{\mu(N-1)\}}{\mu (N-1)}$ we can rewrite \eqref{eq:right_n}
\[
 N+n=\mu(N-1)(1-y + l). 
\]
Substituting  this expression into \eqref{eq:S_N^k} we obtain
\begin{equation*}
\begin{split}
S_1(\angle,N)=\ln \mu(1-(y-l)) +\dfrac{1}{2(N-1)}\left(\dfrac{1}{\mu(1-(y-l))  }-1\right) \\
-\dfrac{1}{12(N-1)^2}\left(\dfrac{1}{\mu^2(1-(y-l))^2}-1\right)+O\left(\dfrac{1}{N^4}\right).
\end{split}
\end{equation*}
Now, we expand the expression for $S_1$ in the perturbation series using

\[\ln(\mu(1-(y-l)))=\ln \mu - \sum\limits_{j=1}^{\infty}\frac{(y-l)^j}{j},\qquad
\frac{1}{\mu(1-(y-l))}=\frac{1}{\mu}\sum\limits_{j=0}^\infty (y-l)^j,\]
and
\[
\frac{1}{\mu^2(1-(y-l))^2}=\frac{1}{\mu^2}\sum\limits_{j=0}^\infty (j+1)(y-l)^j.
\]
Collecting all the terms of the same order in $N$ up to $O(y)$ and $O(N^{-3})$
\begin{equation}
\label{eq:S1}
\begin{split}
S_1(\angle,N)=\ln \mu&+\dfrac{1}{2(N-1)}\left(\frac{1+ 2\gamma -2\{\mu(N-1)  \}  }{\mu}-1\right) \\
&-\frac{1}{12(N-1)^2}  \left (  \dfrac{6\gamma (\gamma+1)}{\mu^2} - 1 \right )   +O\left(y,N^{-3}\right).\end{split}
\end{equation}
We need to choose $\gamma$ in such a way that $\A S_1 <1$ but 
$\A S_1 + \frac{1}{N+n+1}>1$.

Recalling that $\A=\dfrac{1}{\ln\mu}$ and that  the leading contribution for 
deviation from 1, is controlled by terms of order $1/(N-1)$, we must assure
\[
1+ 2\gamma -2\{\mu(N-1)\} < \mu,
\]
but
\[
1+ 2\gamma -2\{\mu(N-1) \} + 2 > \mu.
\]
The additional summand 2 comes from
\[
\dfrac{1}{N+n+1} = \dfrac{1}{[\mu(N-1)] + \gamma +1} = \dfrac{2}{2\mu (N-1)}  + O(1/N^2).
\]

We rewrite two inequalities in a more compact form
\[
1< \mu + 2\{\mu(N-1)\}  -2\gamma <3.
\]
Now, if $\gamma =0$, then clearly left side of the inequality holds, since $\mu>1$, and then right side would also hold provided  $\mu + 2\{\mu(N-1)\} < 3$.
Thus, $\gamma =0$ if $\mu + 2\{\mu(N-1)\} < 3$.

If, on the other hand, $\mu + 2\{\mu(N-1)\} > 3$, we can take $\gamma =1$ as the left side of the inequality will hold. The right side will also hold if apply the assumption
 $\mu <3$. Therefore, we can finally conclude $\gamma = \chi_{\geqslant 0}(\mu + 2\{ \mu(N-1) \}-3).$

\end{proof}

Now, we will derive an explicit expression for the first return map $T$. Multiply \eqref{eq:S1} by $2\A$ 
\begin{equation}
\label{eq:S_positive}
\begin{split}
2\A S_1(\angle,N)&=2+\dfrac{\A}{N-1}\left( \dfrac{1+2\H_\mu(N-1)}  
{\mu}-1\right)-\\ &-\dfrac{\A}{6(N-1)^2}\left(\dfrac{6(\, \chi_{\geqslant 0}(N-1) +1)\, \chi_{\geqslant 0}(N-1)}{\mu^2}-1\right)+O\left(y,N^{-3}\right).
 \end{split}\end{equation}
 So we get for $\psi(\angle, N)$
\begin{equation*}
\begin{split} 
&\A\psi(\angle, N)=2 \A S_1(\angle,N) +\dfrac{\A}{\mu(N-1)+\H_\mu(N-1)+1}-2
=\\&=
\dfrac{\A}{N-1}\left( \dfrac{1+2\H_\mu(N-1) }  
{\mu}-1\right)-\dfrac{\A}{6(N-1)^2}\left(\dfrac{6(\, \chi_{\geqslant 0}(N-1) +1)\, \chi_{\geqslant 0}(N-1)}{\mu^2}-1\right)+ \\
&+\dfrac{\A}{\mu(N-1)}\left(1 - \dfrac{\H_\mu(N-1) +1 }{\mu(N-1)}\right)+ O\left(y,N^{-3} \right).
\end{split}
\end{equation*}

\vspace{2mm}

And finally (using that $\chi_{\geqslant 0}^2 = \chi_{\geqslant 0}$), we have 

\begin{equation}
\label{eq:Kesten}
\begin{split}
\angle'=\angle+\alpha\psi(\angle,\action)=& \angle + \frac{\A}{N-1} \left(\dfrac{ 2 + 2\H_{\mu}(N-1)}{\mu}-1\right)-\\
-&\dfrac{\A}{(N-1)^2}\left(\dfrac{\H_\mu(N-1)+2\, \chi_{\geqslant 0}(N-1) +1}{\mu^2}+ \dfrac{1}{6}\right)+ O\left(y,N^{-3}\right).
\end{split}
\end{equation}

Multiplying both sides of the equation \eqref{eq:Kesten} by the scaling factor $\dfrac{N-1}{\A}$ we obtain
\begin{equation}
\label{eq:Kest_renorm}
\begin{split}
\dfrac{N-1}{\A}\angle'=\tilde{\angle}+&   \left(\dfrac{ 2 + 2\H_\mu(N-1)}{\mu}-1\right)-\\
-&\dfrac{\A}{N-1}\left(\dfrac{\H_\mu(N-1)+2\, \chi_{\geqslant 0}(N-1) +1}{\mu^2}- \dfrac{1}{6}\right)+
O\left(y,N^{-2}\right).
\end{split}
\end{equation}

Since $\angle\in \I_+$ we can use  \eqref{eq:rescaling} 
\[
\tilde{\angle}'=\dfrac{N}{\A}\angle'=\left(1+\dfrac{1}{N-1}\right)\dfrac{N-1}{\A}\angle'.
\] 
Thus, in the leading order in the rescaled variables, the transformation 
takes the form

\vspace{2mm}

\[
\tilde{\angle}'=\tilde{\angle}  +   \dfrac{2}{\mu}(1 + \H_\mu(N-1)) -1
\]

\vspace{2mm}

{\bf Case $\angle\in \I_-$.} 
By the direct calculations, one obtains 
\[
\tilde{\angle}'=\tilde{\angle}  + \dfrac{2}{\mu}(1 + \H_\mu(N-2)) -1.
\]
Indeed, moving in the opposite direction, we have 
\[
\frac{\A}{N-1} -\angle = \frac{\A}{N-2} -\angle' + \frac{\A}{N-2} \left(\dfrac{ 2 + 2\H_ \mu(N-2)}{\mu}-1\right).
\]
Neglecting terms of order $O(N^{-2})$ and rearranging the terms , we obtain
\[
\angle' = \angle + \frac{\A}{N-2} \left(\dfrac{ 2 + 2\H_\mu(N-2)}{\mu}-1\right).
\]
Multiplying with $(N-2)/\A$, 
\[
\tilde{\angle}'=\tilde{\angle}  \, \frac{N-2}{N-1} + \left(\dfrac{ 2 + 2\H_ \mu(N-2)}{\mu}-1\right)
\]
we arrive at the above formula since there we neglect terms of order $O(N^{-1})$.

{\bf Case $\angle\in \I_0$.} 
We already know that $\I_0$ is either a single interval with $\I_+$ and $\I_-$ complementing it 
to the full fundamental interval or $\I_0$ consists of three intervals. In the latter case, $\I_+$ and $\I_-$ are located in the interior of the fundamental domain, while three intervals of $\I_0$ complement at the left boundary, right boundary and in the middle between $\I_+$ and $\I_-$.

First, consider the case of the middle interval of $\I_0$ which can be the whole set $\I_0$ or part of it, as explained above. Compared to the $\I_+$ case, $n$, the number of steps in the positive side of the cylinder decreases by 1, {\em i.e.}
\[
n  = \mu(N-1) - N + \H_\mu(N-1)-1.
\]
Next, we follow the previous calculation of the transformation $T$ on $\I_+$ indicating the changes.
First, recall that on $\I_0$, we always have
\[
\angle' = \angle + 2\A S_1(\angle,I) -2.
\]
Thus, we can modify \eqref{eq:S_positive} for our case by subtracting 1 from $\H$
\[
2 \A S_1 = 2 + \frac{\A}{N-1} \left (   \frac{1 + 2 (\H_\mu(N-1) -1)}{\mu} -1 \right ).
\]
Next using $\A \psi = 2\A S_1 -2$, we have
\[
\angle' = \angle + \frac{\A}{N-1} \left (   \frac{1 + 2 (\H_\mu(N-1) -1) }{\mu} -1 \right )
\]
and renormalizing 
\[
\tilde{\angle}'=\tilde{\angle} + \frac{1 + 2 (\H_\mu(N-1) -1)}{\mu} -1.
\]
After slight simplifications, we finally have
\[
\tilde{\angle}'=\tilde{\angle} + \frac{2}{\mu} ( \H_\mu(N-1) -1/2) -1.
\]

Now, we consider the leftmost subinterval of $\I_0$, if it exists. By previous arguments, $n$ is the same as for the $\I_+$, so we have 
\[
2 \A S_1 = 2 + \frac{\A}{N-1} \left (   \frac{1 + 2 \H_\mu(N-1)}{\mu} -1 \right ).
\]
Using again $\A \psi = 2\A S_1 -2$ and following the same calculations, we find
\[
\tilde{\angle}'=\tilde{\angle} + \frac{2}{\mu} ( \H_\mu(N-1) +1/2) -1,
\]

The last case is the rightmost subinterval of $\I_0$ if it exists. Similar argument leads to
\[
\tilde{\angle}'=\tilde{\angle} + \frac{2}{\mu} ( \H_\mu(N-2) +1/2) -1.
\]

\end{proof}

\subsection{Proof of Theorem \ref{th: escape} }
\begin{proof}
If  $\A = \dfrac{1}{\ln n}$ then  $\mu=n$ and therefore $y\equiv 0$. Thus, all the terms in series expansion \eqref{eq:Kest_renorm} depending on $y$ 
vanish and we get only $O(N^{-3})$ terms. In the leading order, the rescaled transformation takes the form 
\[
\tilde{\angle}'=\tilde{\angle}+\left(\dfrac{2\left(1+\left[\frac{\mu-1}{2}\right]\right)}{\mu}-1\right),
\]
for $\tilde{\angle}\in \I_+$.

For our particular choice $\mu=2m$ we have $\left[\frac{\mu-1}{2}\right]=m-1$ and so in the leading order the transformation $\trans$ is the identity map. To simplify an analysis consider only the case $\mu=\ln 2$. Other cases can be treated similarly. Thus, we get
\[\tilde{\angle}'=\tilde{\angle}+\frac{1}{8(N-1)}-\frac{1}{4N-2}+O(N^{-2})\] 

To construct  an unbounded orbit, we search for an initial  point $\tilde{\angle}_0$ satisfying two conditions: 
\begin{itemize}
\item $\chi_+(\tilde{\angle}_0)=1$,
\item The image $\tilde{\angle}'_0$ under the transformation $T$ coincides with initial angle up to the new normalization $\tilde{\tilde{\angle}}_0$. 
\end{itemize}

Note that for a larger action variable the scaling factor would be different. Therefore, we have 
\[\tilde{\tilde{\angle}}_0=\dfrac{N-1}{N}\tilde{\angle}_0=\tilde{\angle}_0-\frac{1}{N}\tilde{\angle}_0\]
and the ``renormalized fixed point'' condition yields
\[\tilde{\angle}_0-\frac{1}{N}\tilde{\angle}_0=\tilde{\angle}_0+\frac{1}{8(N-1)}-\frac{1}{4N-2}+O(N^{-3}).\] 
Thus, up to the order $O(N^{-3})$
\[-\frac{\tilde{\angle}_0}{N}=\frac{4N-2-8N+8}{8(N-1)(4N-2)}=-\frac{2N-3}{8(N-1)(2N-1)}=\]\[=-\frac{1}{8(N-1)}+\frac{1}{4(N-1)(2N-1)}=-\frac{1}{8N}-\frac{1}{8N(N-1)}+\frac{1}{4(N-1)(2N-1)}=\]
\[=-\frac{1}{8N}+\frac{1}{4(N-1)}\left(\frac{1}{2N}-\frac{1}{2N-1}\right)=-\frac{1}{8N}-\frac{1}{8N(N-1)(2N-1)}.\]
As a result,  for $\tilde{\angle}_0=\dfrac 18+O(N^{-3})$ its image under first return map is given by 
 $\tilde{\angle}'_0=\dfrac 18 + O(N^{-3})$.

To check the first statement  $\chi_+(\angle_0)=1$ it is sufficient to show that $\angle_0+2\A S_1+\dfrac{\A}{N+n+1}$ is greater than $2$. This immediately follows from \eqref{eq:S_positive} and the estimate
 \[\angle_0+2\A S_1+\dfrac{\A}{N+n+1}>\dfrac{1}{8N\ln 2}+2-\dfrac{1}{\ln 2(2N-2)}+O(N^{-2})+\dfrac{1}{(2N-1)\ln 2}>2.\]
\end{proof}

\begin{figure}[ht]
\centering
\vspace{-1.3 in}
\includegraphics[width= 4in]{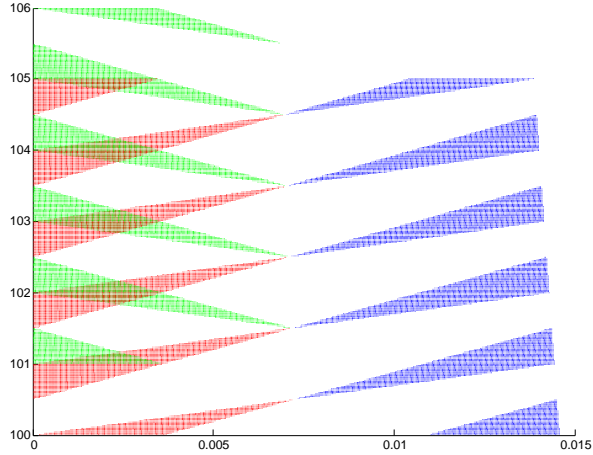}
\vspace{-1.5 in}
\caption{The structure of the fundamental domain for $\A=\dfrac{1}{\log 2}$.  The regions of  positive growth are in red and their images are in green. The regions of negative growth are in blue. }
\end{figure}

\begin{remark}
For 
$\mu= 2m+1$, $m\in \mathbb{N}$ rotation number for the map equals $(\mu-1)/\mu$ which guarantees that any trajectory could not hit $\I_+$ two successive times. Moreover such map has a fixed point in the set $\I_0$. However since renormalization for the neutral set is an identity map, drift produced by the second order term could not be compensated. We conjecture that in this case one could possibly derive an exponential bounds on the rate of action growth for the map $\trans$.     
\end{remark}

\begin{remark}
We conjecture also that  
the same arguments as in Lemmas \ref{lm: First_return} -- \ref{lm: renormalization_Kesten}  can be used for other values of $\Z\in (-1,0)$. Here instead of asymptotic expansion for harmonic numbers one should use  generalized harmonic numbers $H_{n,\Z}$. Since $H_{n,\Z}\to \zeta(k) $ as $n\to\infty$ one can establish the  relation for $H_{(\action+n),\Z}-H_{(\action-1),\Z}=2$ to find the expression for $n$. Arguments in Lemma \ref{lm: First_return}, \ref{lm:positive-negative} and \ref{lm: renormalization_Kesten} can be 
applied with minor changes.
\end{remark}

\appendix
\section{}

In the appendix we give more detailed derivation of some specific problems that lead to $\A\Z$ map.

\subsection*{Example 1: Particle in switching potential}  

Consider a classical particle moving on the line in the square wave potential $V(x) = (-1)^{[x]}$ and assume 
the potential is switched every time unit $V(x,t) = (-1)^{[t]}  \cdot (-1)^{[x]}$. While such potential 
is not differentiable, there is a natural way to define the dynamics by using the energy relation: the kinetic 
energy changes by 2 if the particle passes $t\in {\mathbb Z}$ integer points. We should ignore the singular subset of the extended phase space $(t,\dot x, x)$ where there is discontinuity in both time and space and the dynamics 
is not defined. Such subset has zero measure. Outside the singular set
particle moves with constant speed $v(x,t)=\sqrt{E-V(x,t)}$.
 
 Since the dynamics defined for $x \mod{2}$, $t\mod{2}$ we can project the dynamics on a plane onto the system on a cylinder $[0,2)\times \mathbb{R}_+$. Write the Hamiltonian form of the unit time step transformation for this system. Hamiltonian has a form $H(x,\dot{x},t)=\frac{\dot{x}^2}{2}+V'(x,t)$ and so for the canonical action-angle variables we get
\begin{equation}
\label{eq:Potential_Hamiltonian}
\begin{cases}
\action_{\pm}=\sqrt{2(H-1)}+\sqrt{2(H+1)}\\
\angle_{\pm}=\dfrac 1T\left(\dfrac{x-1}{\sqrt{2(H\pm\sgn(x-1))}}+\dfrac{1}{\sqrt{2(H\mp 1)}}\right)
\end{cases}
\end{equation}
Here  
$T=\left(\dfrac{1}{\sqrt{2(H-1)}}+\dfrac{1}{\sqrt{2(H+1)}}\right)$ is a period of rotation along the level curve $H(x,\dot{x})=H$ and $(\angle_{\pm}, \action_{\pm})$ correspond to the angle-action variables for odd/even values of $[t]$.
To deduce the system in action-angle variables, write down a generating function.
For Hamiltonian we have

\[H(x,\dot{x},t)=H_0(\action,t)+\partial_3 S_{\pm}(x(t),\action,t)\]
Where
\[\begin{cases}
S_+(x,\action)=(x-1)\sqrt{2(\action+\sgn(x-1))}+\sqrt{2(\action-1)}
\\
S_-(\action)=iS_+(-\action)
\end{cases}\]
Or, using \eqref{eq:Potential_Hamiltonian}
\[S_+=\begin{cases}
2T(\action+1)\varphi_+-\dfrac{4}{\sqrt{2(\action-1)}},\, x>1\\
2T(\action-1)\varphi_+,\,x<1
\end{cases}\]
Finally, the total system in variables $(\angle, \action)$ can be deduced from the expression \eqref{eq:Potential_Hamiltonian} and the relation $\varphi_{\pm}=\partial_I S_{\pm}$.
 
 \begin{equation}
\label{eq:AZ_energ}
\begin{cases}
\angle_1=(\angle+\A \sqrt{\action})\mod{2}
\\
\action_1=\action+\sgn(\angle_1-1)
\end{cases}
\end{equation} 
where $\A$ is some constant. 
Clearly, this system can be considered as a particular example of  transformation \eqref{eq:AZ} with $\Z=\frac 12$. 

\subsection*{Example 2: Fermi-Ulam acceleration.}
In Fermi-Ulam problem the particle bounces between two walls. Assume that one wall is at rest $x=0$ and 
the other moves periodically $x=p(t)$, $p(t+1)=p(t)>0$. There is a standard transformation ``stopping'' the wall,
see e.g. \cite{myself} with 
\[
x = p(t) y, \,\,\, \tau = \int_0^t \frac{ds}{p^2(s)} ds.
\] 
In the new variables, the equation takes the form
\[
y^{\prime \prime} + \ddot p p^3 y = 0,
\]
where $\prime$ denotes the derivative with respect to $\tau$. Evaluating the one period map, under the assumption that $p$ is piecewise linear,
we  obtain
\begin{equation}
\begin{cases}
y_2 = ( y_1 + y_1^{\prime} ) \mod{1}  \\
y_2^{\prime} = y_1^{\prime} + y_2 \,  \sgn \, (y_2-1/2).
\end{cases}
\end{equation}

This mapping is not a particular case of $\A\Z$-map but it corresponds to the linear growth 
of the action when $\Z=1$. While showing unbounded growth is relatively easy in this case,
more challenging problem is to estimate the relative measure of bounded solutions. 
This has been done in \cite{DolSim}.

\subsection*{Example 3: Outer billiards with degenerate boundary.}

Consider the outer billiard system. Let $\gamma$ be a smooth strictly convex curve on the plane. Take any point $x\in \mathbb{R}^2$ outside of $\gamma$ and let $l(x)$ be a ray tangent to $\gamma$  and oriented in the counter-clockwise direction. There is another point on the ray $T(x)$ which has the same distance to the tangency point as $x$. This defines the outer billiard map, see e.g.  \cite{serezha}.  A natural question is whether all the orbits are bounded. Thus, one is led  to study this map for large $x$. 

Assume that $\gamma$ is a unit circle centered at the origin, then for large $x$ the square of the map $T^2$ is close to identity and it leaves concentric circles invariant. The angle changes by a factor of $1/|x|$, which corresponds to $\Z=-1$ in the $\A\Z$-map. Now, consider the circle with a small circle segment removed.  Then, the map becomes a small discontinuous perturbation of the above integrable map. When, $\gamma$ is half the circle,
the outer billiard has unbounded orbit, see \cite{genin, DolFay}. 

Earlier, Schwartz \cite{schwartz} constructed unbounded orbits for quadrilateral $\gamma$ . 

\subsection*{Aknowlegments.}  First author (M.A.) was supported by AFOSR MURI grant FA9550-10-1-0567. The research of V.Z. was partially supported by NSF DMS-0807897.

\bibliographystyle{plain}
\bibliography{pinball}

\end{document}